\documentclass[12pt]{article}
\usepackage{natbib}            
\usepackage{graphicx}          
\usepackage{graphicx}%
\usepackage{amsfonts,amsmath,amssymb}

\newtheorem{theorem}{Theorem}
\newtheorem{acknowledgement}[theorem]{Acknowledgement}

\newtheorem{corollary}[theorem]{Corollary}

\newtheorem{lemma}[theorem]{Lemma}

\newtheorem{problem}[theorem]{Problem}

\newtheorem{remark}[theorem]{Remark}

\newenvironment{proof}[1][Proof]{\noindent\textbf{#1.} }{\ }
\newenvironment{keywords}[1][Keywords]{\noindent\textbf{#1:} }{}

\def\rbb{\mathbb{R}}
\def\trp{^T}
\def\diag{{\rm diag}}
\def\tr{\mathop{\rm Tr}\nolimits} 

\def\half{\frac{1}{2}}
\begin{document}

\title{Coherent Quantum LQG Control \thanks{This work was supported by the Australian Research Council. 
A preliminary version of this paper appeared in the Proceedings of
the 17th IFAC World Congress, 2008.}}  

\author{Hendra~I.~Nurdin$^{\ast}$, Matthew~R.~James$^{\ast}$, and Ian R.
Petersen$^{\dag}$ \and \{Hendra.Nurdin,Matthew.James\}@anu.edu.au,
i.petersen@adfa.edu.au
\medskip \\$^{\ast}$ Department of Engineering \\ The Australian
National University \\ Canberra, ACT 0200, Australia \medskip
\\$^{\dag}$School of
    Information Technology and Electrical Engineering\\
    University of New South Wales at the Australian Defence Force Academy\\
    Canberra, ACT 2600, Australia}

\date{}
\maketitle

\begin{abstract}
Based on a recently developed notion of physical realizability for
quantum linear stochastic systems, we formulate a quantum LQG
optimal control problem for quantum linear stochastic systems
where the controller itself may also be a quantum system and the
plant output signal can be fully quantum. Such a control scheme is
often referred to in the quantum control literature as ``coherent
feedback control.'' It distinguishes the present work from previous works
on the quantum LQG problem where measurement is performed on the
plant and the measurement signals are used as input to a fully
classical controller with no quantum degrees of freedom. The
difference in our formulation is the presence of additional
non-linear and linear constraints on the coefficients of the
sought after controller, rendering the problem as a type of
constrained controller design problem. Due to the presence of
these constraints our problem is inherently computationally hard
and this also distinguishes it in an important way from the
standard LQG problem. We propose a numerical procedure for solving
this problem based on an alternating projections algorithm and, as
initial demonstration of the feasibility of this approach, we
provide fully quantum controller design examples in which
numerical solutions to the problem were successfully obtained. For
comparison, we also consider the case of classical linear
controllers that use direct or indirect measurements, and show
that there exists a fully quantum linear controller which offers
an improvement in performance over the classical ones.
\end{abstract}

\begin{keywords}
Quantum systems; Quantum control;  Stochastic control; Linear quadratic regulators; Linear control systems
\end{keywords}

\section{Introduction}
\label{sec:intro} Recent successes in quantum and nano-technology
have provided a great impetus for research in the area of quantum
feedback control systems; e.g. see
\cite{VPB83,WM93a,DJ99,AASDM02,GSM04}. It is reasonable to expect
that quantum control is an area of research which could play a
vital role towards realization of conceptual quantum information
systems and quantum computers which are being extensively studied
for potential benefits over their classical counterparts. One
particular area in which significant theoretical and experimental
advances have been achieved is \emph{quantum optics}. In
particular, \emph{linear quantum optics} is one of the possible
platforms being investigated for future communication systems
(see \cite{JPF02,KWD03}) and quantum computers (see \cite{KLM01} and
\cite[Section 7.5]{NC00}), besides being an area of independent
interest in physics. Interestingly, under appropriate assumptions,
the dynamics of some quantum optical devices can be approximately
modeled by linear quantum stochastic differential equations
driven by quantum Wiener processes; see \cite{GZ00}. For details on
quantum stochastic differential equations and Wiener processes,
see \cite{HP84,KRP92,BvHJ07}.

In general, quantum linear stochastic systems represented by
linear Quantum Stochastic Differential Equations (QSDEs) with
arbitrary constant coefficients need not correspond to physically
meaningful systems. This is the same as for classical stochastic
differential equations (throughout this paper we shall use the
term ``classical'' to loosely refer to systems that have no
quantum mechanical components). However, because classical linear
stochastic systems can be implemented at least approximately,
using analog or digital electronics, we regard them as always
being realizable. Physical quantum systems must satisfy some
additional constraints that restrict the allowable values for the
system matrices defining the QSDEs. In particular, the laws of
quantum mechanics dictate that closed quantum systems evolve
\emph{unitarily}, implying that (in the Heisenberg picture)
certain canonical observables satisfy the so-called canonical
commutation relations (CCR) at all times. Therefore, to
characterize physically meaningful systems,  \cite{JNP06} has
introduced a formal notion of physically realizable quantum linear
stochastic systems and derives a pair of necessary and sufficient
characterizations for such systems in terms of constraints on
their system matrices.

In this work, we build on the ideas in \cite{JNP06} and
\cite{SPJ1a} and formulate a quantum LQG optimal control  problem
for quantum linear stochastic systems. The distinguishing feature
of our work compared to previous treatments of the quantum LQG
problem in the literature is that we allow the controller to be
another quantum system whereas previous works only consider the
case where the controller is a classical system driven by the
result of continuous measurements performed on the output of the
quantum plant. In the physics literature, (feedback) control using
a fully quantum system is often referred to as 
``coherent feedback control'' to distinguish it from
control using a classical controller. Coherent controllers are of
interest for, among other things, their potential for providing
faster speed of processing/higher bandwidth (by dispensing of the
use of ``slow'' electronics) and better performance (as will be
demonstrated later in this paper). We stress that the coherent LQG
controller design considered herein leads to a more difficult
problem which cannot be solved using the usual approach of quantum
conditioning and dynamic programming. By viewing the problem as a
polynomial matrix programming problem, we show that by utilizing a
non-linear change of variables due to \cite{SGC97}, the problem
can be systematically converted to a rank constrained LMI problem.
To demonstrate the feasibility of numerically solving this
problem, we provide a design example of stabilization of a quantum
plant for which a solution to the rank constrained LMI problem was
successfully obtained using an alternating projections algorithm
due to \cite{OHM06}.

The organization of the paper is as follows. We begin in Section
\ref{sec:models} with an overview of quantum linear stochastic
systems that are of interest in quantum optics. Section
\ref{sec:phys-real} then recalls some key definitions and results
from \cite{JNP06} on physically realizable systems. In Section
\ref{sec:qlqg-form}, we formulate a novel quantum LQG problem that
allows the controller to be another quantum system. A numerical
procedure for solving the quantum LQG problem is then proposed in
Sections \ref{sec:lmi-rank} and \ref{sec:numerics}. In Section
\ref{sec:extension}, we discuss an extension of the methodology
developed in the preceding two sections. In Section
\ref{sec:examples}, we consider the design of a fully quantum LQG
controller for stabilizing a marginally stable quantum plant and
show that there exists a fully quantum controller which offers an
improved level of performance over a fully classical linear
controller which uses direct or indirect measurements. Finally, in
Section \ref{sec:conclude} we offer some concluding remarks.

\section{General quantum linear stochastic models in quantum optics}
\label{sec:models} We follow the quantum probabilistic setup of
\cite[Section II]{JNP06} and recall the notion of physical
realizability introduced therein. To this end, consider linear
non-commutative stochastic systems of the form
\begin{eqnarray}
dx(t) &=& Ax(t) dt + B dw(t); \quad x(0)=x_0
\nonumber \\
dy(t) &=& C x(t) dt + D dw(t) \label{linear-c}
\end{eqnarray}
where $A$, $B$, $C$ and  $D$   are real matrices in $\rbb^{n
\times n}$, $\rbb^{n \times n_w}$, $\rbb^{n_y \times n}$ and
$\rbb^{n_y \times n_w}$ respectively. Also,  ($n,n_w,n_y$ are positive
integers), and $ x(t) = [\begin{array}{ccc} x_1(t) & \ldots &
x_n(t)
\end{array}]\trp$ is a vector of self-adjoint possibly
non-commutative system variables.

The initial system variables $x(0)=x_0$ are Gaussian with state
$\rho$\footnote{That is, $\tr(\rho
e^{i\lambda^Tx_0})=e^{i\lambda^T m-\frac{1}{2}\lambda^T G
\lambda}$ for all $\lambda \in \rbb^n$ where $m \in \rbb^n$ and
$G$ is a real symmetric matrix satisfying $G+i\Theta \geq 0$ with $\Theta$ as
given in the text above; see, e.g., \cite{Lin98,GZ00,KRP92}).}, and %
satisfy the {\em commutation relations}\footnote{In the case of a
single degree of freedom quantum particle, $x=(x_1, x_2)^T$ where
$x_1=q$ is the position operator, and $x_2=p$ is the momentum
operator. The annihilation operator is $a=(q+ip)/2$. The
commutation relations are $[a,a^\ast]=1$, or $[q,p]=2 i$.}
\begin{equation}
[x_j(0), x_k(0) ] = 2 i \Theta_{jk}, \ \ j,k = 1, \ldots, n,
\label{x-ccr}
\end{equation}
where $\Theta$ is a real antisymmetric matrix with components
$\Theta_{jk}$, and $i=\sqrt{-1}$.  Here, the commutator is defined
by $[A,B]=AB-BA$. To simplify matters and without loss of generality,
we take the matrix $\Theta$ to be of one of the following forms:
\begin{itemize}
\item \emph{Canonical} if $\Theta=\diag(J,J,\ldots,J)$, or

\item \emph{Degenerate canonical} if $\Theta=\diag(0_{n' \times
n'},J,\ldots,J)$, where $0< n' \leq n$.
\end{itemize}
Here, $J$ denotes the real skew-symmetric $2\times 2$ matrix
$$
J= \left[ \begin{array}{cc} 0 & 1 \\ -1 & 0
\end{array} \right],$$
and the \lq\lq{diag}\rq\rq \ notation indicates a block diagonal
matrix assembled from the given entries. To illustrate, the case
of a system with one classical variable and two conjugate quantum
variables is characterized by $\Theta=\mathrm{diag}(0,J)$, which
is degenerate canonical. The vector quantity $w$ describes the
input signals and is assumed to admit the decomposition
\begin{equation}
dw(t) = \beta_w(t)dt + d\tilde w(t) \label{w-general}
\end{equation}
where $\tilde w(t)$ is the noise part of $w(t)$ and $\beta_w(t)$
is a self adjoint, adapted process (see \cite{HP84,KRP92,BvHJ07}
for a discussion of adapted quantum processes). The process
$\beta_w(t)$ serves to represent variables of other systems that
may be passed to the system (\ref{linear-c}) via a connection. It
is also represented on the same quantum probability space, which
can be enlarged if necessary. Consequently, we assume that
components of $\beta_w(t)$ commute with those of $dw(t)$.
Furthermore, we will also assume that components of $\beta_w(t)$
commute with those of $x(t)$; this will simplify matters for the
present work. The noise $\tilde w(t)$ is a vector of self-adjoint
quantum noises with Ito table
\begin{equation}
d\tilde w(t) d \tilde w^T(t) = F_{\tilde w} dt,
\label{w-Ito-general}
\end{equation}
where $F_{\tilde w}$ is a non-negative Hermitian matrix; e.g.,  see
\cite{KRP92,VPB91}. This determines the following commutation
relations for the noise components:
\begin{equation}
[ d\tilde w(t), d\tilde w^T(t) ] = d\tilde w(t) d\tilde w^T(t) -
(d\tilde w(t) d\tilde w^T(t))^T =  2T_{\tilde w}dt , \label{w-ccr}
\end{equation}
where we use the notation $ S_{\tilde w} = \frac{1}{2} ( F_{\tilde
w} + F_{\tilde w}^T),  \ \ T_{\tilde w} = \frac{1}{2} ( F_{\tilde
w} - F_{\tilde w}^T)$ so that $F_{\tilde w} = S_{\tilde w} +
T_{\tilde w}$. For instance, $F_{\tilde w} =
\mathrm{diag}(1,I+iJ)$ describes a noise vector with one classical
component and a pair of conjugate quantum noises. (Here $I$ is the
$2\times 2$ identity matrix.) The noise processes can be
represented as operators on Fock spaces; e.g.,  see
\cite{HP84,KRP92}.

For simplicity, we also adopt the conventions of \cite{JNP06} to
put the system (\ref{linear-c}) into a standard form. Therefore,
we assume that in (\ref{linear-c}): (i) $n_y$ is even, and (ii) $n_w
\geq n_y$. Furthermore, we also assume that $F_{\tilde w}$ is of
the \emph{canonical} form $F_{\tilde w}=I+i\diag(J,\ldots,J)$.
Hence $n_w$ has to be even. Note that if $F_{\tilde w}$ is not
canonical but of the form $F_{\tilde w}=I+i\diag(0_{n' \times
n'},\diag(J,\ldots,J))$ with $n' \geq 1$, we may enlarge $w(t)$
(and hence also $\tilde w(t)$) and $B$ as before such that the
enlarged noise vector, say $\tilde w'$, can be taken to have an
Ito matrix $F_{\tilde w'}$ which is canonical.


\section{Physical realizability of linear QSDEs}
\label{sec:phys-real} A linear quantum stochastic system
(\ref{linear-c}) with arbitrary system matrices need not represent
any meaningful physical system. For example, quantum mechanics
dictates that closed physical quantum systems evolve in a unitary
manner which (in the Heisenberg or interaction picture) implies
the preservation the canonical commutation relations (CCR):
$x(t)x(t)\trp-(x(t)x(t)\trp)\trp=2i\Theta \ \mbox{for all } t\geq 0$.
This essentially restricts the allowable coefficients $A,B,C,D$
for a physically realizable quantum linear stochastic system. For
this reason, \cite{JNP06} develops a precise notion of physical
realizability based around the concept of an \emph{open quantum
harmonic oscillator} \cite[Section 2]{JNP06} as the basic
``dynamical unit'' of a physically realizable quantum system. Such
an oscillator is completely described by a quadratic Hamiltonian
$H=\half x(0)\trp R x(0)$, where $R$ is a real symmetric matrix, and a
linear coupling operator $L=\Lambda x(0)$, where $\Lambda$ is a
complex matrix (here we correct a minor error in 
\cite[Definition 3.1]{JNP06} where the $\half$ factor had been
omitted from the definition of $H$). We now give a brief summary of
the main ideas developed in \cite{JNP06}.

In the case that the system is fully quantum (i.e. $\Theta$ is
canonical) then (\ref{linear-c}) is said to be physically
realizable if it represents the dynamics of some open quantum
harmonic oscillator \cite[Section III-A]{JNP06}. If $\Theta$ is
not canonical then a classical component of $x(t)$, say $x_i(t)$,
is viewed as one component of a pair of (fictitious) canonically
conjugate operators $(x_i(t),z_i(t))$ satisfying
$[x_i(t),z_i(t)]=2i$ and which commute with all other components
of $x(t)$. In that case (\ref{linear-c}) is said to be physically
realizable if it can be embedded in some larger system also of the
form (\ref{linear-c}) with a special structure (referred to as an
\emph{augmentation} of (\ref{linear-c}) in \cite[Section
III-B]{JNP06}) which is itself physically realizable, i.e. the
larger system represents the dynamics of an open quantum harmonic
oscillator. Based on this definition, \cite{JNP06} derives a pair
of necessary and sufficient conditions for a quantum linear
stochastic system of the form (\ref{linear-c}) to be physically
realizable in terms of the system matrices, regardless of whether
$\Theta$ is canonical or degenerate canonical. It is as follows,
with $N_y=n_y/2$, $N_w=n_w/2$ and using the notation $\diag_m(K)$
to denote a block diagonal matrix with a square matrix $K$
appearing $m$ times on its diagonal blocks:

\begin{theorem}[\cite{JNP06}]  \label{thm_phys_real}
The system (\ref{linear-c}) is physically realizable if and only
if:
\begin{eqnarray}
iA\Theta + i\Theta A \trp + BT_wB\trp&=&0
\label{eq_realize_cond_A}, \\
B\left[\begin{array}{c} I_{n_y \times n_y} \\
0_{(n_w-n_y) \times n_y}
\end{array} \right]&=&
\Theta C\trp\diag_{N_y}(J) \label{eq_realize_cond_B},
\end{eqnarray}
and $D=[\begin{array}{cc}I_{n_y \times n_y} & 0_{n_y \times
(n_w-n_y)}\end{array}]$. Moreover for canonical $\Theta$, the
corresponding Hamiltonian and coupling matrices have explicit expressions as
follows. The Hamiltonian matrix $R$ is uniquely given by
$R=\frac{1}{4}(-\Theta A+A\trp \Theta)$, and the coupling matrix
$\Lambda$ is given uniquely by
\begin{eqnarray}
\Lambda=-\frac{1}{2}i\left[\begin{array}{cc} 0_{N_w \times N_w} &
I_{N_w \times N_w}\end{array} \right] (\Gamma^{-1})\trp B\trp
\Theta,
\end{eqnarray}
where $\Gamma=P_{N_w}\diag_{N_w}(M)$ with
$M=\frac{1}{2}\left[\begin{array}{cc} 1 & i\\1 &
-i\end{array}\right]$ and $P_{N_w}$ is a permutation matrix acting
as
$P_{N_w}(a_1,a_2,\ldots,a_{2N_w})^T=(a_1,a_3,\ldots,a_{2N_w-1},a_2,a_4,\ldots,a_{2N_w})^T$.
In the case that $\Theta$ is degenerate canonical, a physically
realizable augmentation of the system can be constructed  to
determine the Hamiltonian and coupling operators using the
explicit formulas above.
\end{theorem}

Note that background information concerning the Hamiltonian and
coupling matrices in quantum systems can be found in references
\cite{EB05,GZ00,JNP06}.

\section{Formulation of the quantum LQG problem}
\label{sec:qlqg-form} We consider {\em plants} described by
non-commutative stochastic models of the following form:
\begin{eqnarray}
dx(t) &=& Ax(t)dt + B du(t) + B_w dw(t); \quad x(0)=x; \nonumber \\
dy(t) &=& Cx(t)dt +D_w dw(t); \nonumber \\
z(t) &=& C_z x(t)+ D_z \beta_u(t). \label{sys}
\end{eqnarray}
Here $x(t)$ is a vector of plant variables, $w(t)$ is a quantum
Wiener disturbance vector, $\beta_u(t)$ is an adapted,
self-adjoint process commuting with $x(t)$ (i.e.  $\beta_u(t)
x(t)\trp-(x(t)\beta_u(t)\trp)\trp=0$), and $u(t)$ is a {\em
control} input of the form
\begin{equation}
du(t) = \beta_u(t)dt + d\tilde u(t) \label{u-quant}
\end{equation}
where $\beta_u(t)$ is the ``signal part'' and $\tilde u(t)$ is the
noise part of $u(t)$. The vectors $w(t)$ and $\tilde u(t)$ are
independent quantum noises (meaning that they live on distinct
Fock spaces) with Ito matrices $F_{w}$ and $F_{\tilde u}$ that are
all non-negative Hermitian. We also assume that
$x(0)x(0)\trp-(x(0)x(0)\trp)\trp=\Theta.$

{\em Controllers} are assumed to be non-commutative stochastic
systems of the form
\begin{eqnarray}
d\xi(t) &=& A_K\xi(t)dt +B_{K1}dw_{K1}(t)+ B_{K2}dw_{K2}(t)+ B_{K3}dy(t); \nonumber \\
du(t) &=& C_K\xi(t)dt + dw_{K1}(t) \label{controller}
\end{eqnarray}
where $ \xi(t) = [\begin{array}{ccc} \xi_1(t) & \ldots &
\xi_{n_K}(t)\end{array}]\trp $ is a vector of self-adjoint
controller variables of the same dimension as $x(t)$ (i.e.  the
controller is of the same order as the plant), $B_{K2}$ is a
square matrix of the same dimension as $A_K$, and $B_{K1}$ has the
same number of columns as there are rows of $C_K$. The noises $
w_{Ki}(t)$, $i=1,2$, are vectors of non-commutative Wiener
processes (in vacuum states) with non-zero Ito products and which
are independent of $w(t)$. We assume that
$\xi(0)\xi(0)\trp-(\xi(0)\xi(0)\trp)\trp=\Theta_K$. Here
$\Theta_K$ is the skew-symmetric commutation matrix for the
controller variables $\xi$ that could be of a canonical or
degenerate canonical form (cf. Section \ref{sec:models}).

Assume further that $x(0)\xi(0)\trp-(\xi(0)x(0)\trp)\trp=0$, i.e.
the plant and controller are initially decoupled. The {\em closed
loop system} is obtained by the identification $\beta_u(t)\equiv
C_K\xi(t)$ and $\tilde u(t)\equiv w_{K1}(t)$, and interconnecting
(\ref{sys}) and (\ref{controller}) to give
\begin{eqnarray}
\label{cl1} d\eta(t) &=& \mathcal A \eta(t)  dt + \mathcal B
dw_{cl}(t); \nonumber \\
z(t) &=& \mathcal C \eta(t) \label{eq:closed-loop}
\end{eqnarray}
where $\eta(t) = [\begin{array}{cc}x(t)\trp & \xi(t)\trp
\end{array}]\trp$,
\begin{eqnarray*}
w_{cl}(t) &=& \left[\begin{array}{l} w(t) \\ w_{K1}(t) \\
w_{K2}(t)
\end{array}
\right];~~\mathcal{A} = \left[\begin{array}{ll}A & B C_K
\\B_{K3} C  & A_K
  \end{array}\right]; \nonumber \\
\mathcal B &=& \left[\begin{array}{lll}B_w & B & 0_{2\times 2}
\\B_{K3}D_w & B_{K1} & B_{K2}\end{array}\right]; \quad
\mathcal C = \left[\begin{array}{ll}C_z &
    D_z C_K\end{array}\right].\end{eqnarray*}

With (\ref{eq:closed-loop}) we associate a quadratic performance
index
\begin{equation}
J(t_f) = \int_0^{t_f} \langle   z^{T}(t) z(t) \rangle ~dt .
\label{LQG-index-1}
\end{equation}
Here the notation $\langle \cdot \rangle$ is standard and refers
to quantum expectation (e.g., see \cite{EM98}). In this case, the
quantum expectation is on the composite system (plant, controller and
all quantum noises) following \cite{JNP06}.

We shall proceed to derive an explicit expression for this performance index;
see also \cite{SPJ1a}. To this end, define the symmetrized
covariance matrix $P(t)$ by
\begin{equation}
P(t)= \frac{1}{2}\langle  \eta(t) \eta^{T}(t) +   ( \eta(t)
\eta^{T}(t))^T \rangle.
\end{equation}
Using the quantum Ito rule, we have
\[
\begin{array}{rcl}
d P(t) &=& \frac{1}{2} ( \langle  d\eta(t) \, \eta^{T}(t) \rangle
+ \langle (d\eta(t) \, \eta^{T}(t))^T  \rangle + \langle \eta(t)
\, d\eta^{T}(t)  \rangle + \\
&& \langle (\eta(t) \, d\eta^{T}(t))^T \rangle+ (\mathcal B
F_{w_{cl}} \mathcal B^{T} + (\mathcal B F_{w_{cl}}
\mathcal B^{T})^T )~dt) \\[2mm]
&=& (\mathcal  A   P(t)  +  P(t)\mathcal
 A^{T} +  \frac{1}{2}\mathcal  B  (F_{w_{cl}} +F_{w_{cl}}^T)
\mathcal B^{T})~dt,\\[2mm]
&=& (\mathcal  A   P(t)  +  P(t)\mathcal
 A^{T} + \mathcal  B \mathcal B^{T})~dt,
\end{array}
\]
where the last equality follows from our convention that all
noises are canonical (hence $\frac{1}{2}(F_{w_{cl}}
+F_{w_{cl}}^T)=I$). Hence $P(\cdot)$ satisfies the differential
equation
\begin{equation} \label{DRE}
\dot{P}(t)= \mathcal A  P(t)+  P(t) \mathcal A^{T}+ \mathcal B
\mathcal B^{T}; \quad P(0)=P_0.
\end{equation}
We now have, using the symmetry of $\mathcal C^{T} \mathcal C$ and
$P$, that
\[
\begin{array} {rcl}
\langle  z^{T} z  \rangle &=& \langle  \eta^{T}
\mathcal C^{T} \mathcal C\eta \rangle \\[2mm]
&=&\langle  \mbox{Tr} (\eta^{T}
\mathcal C^{T} \mathcal C \eta) \rangle \\[2mm]
&=& \frac{1}{2}\langle  \mbox{Tr} (
\mathcal C^{T} \mathcal C [\eta \eta^{T} +(\eta \eta^{T})^T]) \rangle \\[2mm]
&=&  \mbox{Tr} ( \mathcal C^{T} \mathcal C P).
\end{array}
\]
Hence, the performance index (\ref{LQG-index-1}) can be expressed as
\begin{equation} \label{Cost}
J(t_f)=\int_0^{t_f} \mbox{Tr} ( \mathcal C^{T} \mathcal C P(t))
~dt
\end{equation}
where $P(t)$ solves (\ref{DRE}). We will focus our attention on
the infinite horizon case where we allow $t_f \uparrow \infty$.
Assuming that $\mathcal A$ is asymptotically stable, standard
results on Lyapunov equations give us $ \lim_{t\to\infty} P(t) =
P$, where $P$ is the unique symmetric positive definite solution
of the Lyapunov equation:
\begin{equation}
\mathcal AP+ P \mathcal A^{T}+ \mathcal B \mathcal B^{T} =0.
\label{LQG-lyapunov}
\end{equation}
Furthermore, by standard methods of analysis we have
\begin{equation*}
\limsup_{t_f \to \infty} \frac{1}{t_f} \int_0^{t_f} \langle z^T(t)
z(t) \rangle dt = \mathrm{Tr} ( \mathcal C^{T} \mathcal C P)=
\mathrm{Tr} (\mathcal C P  \mathcal C^{T}). 
\end{equation*}

As before, let $\diag_{m}(J)$ denote a block diagonal $2m \times
2m$ matrix with $m$ $J$ matrices on the diagonal blocks and let
$n_i$ denote the dimension of $w_{Ki}$ for $i=1,2,3$. We may now
formulate our quantum LQG control problem  for
an infinite horizon as follows:

\begin{problem}[Quantum LQG synthesis]
\label{pb:lqg-main} Given a fixed choice of $\Theta_K$, find
controller matrices $A_K$, $B_{K1}$, $B_{K2}$, $B_{K3}$ and $C_K$
that minimizes the cost functional $J_{\infty}=\tr(\mathcal C P
\mathcal C\trp)$ subject to the constraint that the controller
(\ref{controller}) is physically realizable. That is,
$A_K$, $B_{K1}$, $B_{K2}$, $B_{K3}$, $C_K$ satisfy the conditions of Theorem
\ref{thm_phys_real} with the identification: $A \equiv A_K$,
$B=[\begin{array}{ccc} B_{K1} & B_{K2} & B_{K3}
\end{array}]$, $C \equiv C_K$, $D \equiv [\begin{array}{cc} I_{n_u \times n_u} & 0 \end{array}]$, and
$w \equiv [\begin{array}{ccc} w_{K1} \trp & w_{K2}\trp &
y\trp\end{array}]\trp$, leading to the pair of constraints:
\begin{eqnarray}
&&A_K \Theta_K + \Theta_K A_K\trp + B_{K1}\,\diag_{n_1/2}(J)\,
B_{K1}\trp + B_{K2}\,\diag_{n_2/2}(J)\,B_{K2}\trp + \nonumber \\
&&\quad B_{K3}\,\diag_{n_3/2}(J)\,B_{K3}\trp=0; \label{eq:cntrl-pr-1}\\
&&B_{K1}=\Theta_K C_K\trp \diag_{n_u/2}(J) \label{eq:cntrl-pr-2}.
\end{eqnarray}
\end{problem}

In the above problem $\Theta_K$ is a fixed but freely specified
parameter that determines the type of controller sought. For
example, if $\Theta_K$ is canonical then the controller will be
fully quantum. Our formulation of the quantum LQG problem
\emph{differs} from previous formulations of the quantum LQG
problem, such as given in \cite{EB05} and \cite{DJ99}. The
important difference is that in the earlier works, the controller
is classical whereas in our formulation we seek a controller which
may possibly be another quantum system (depending on how
$\Theta_K$ is defined) which generates an optical field to drive
the quantum plant. What is new in the formulation are the
additional constraints (\ref{eq:cntrl-pr-1}) and
(\ref{eq:cntrl-pr-2}) that must also be satisfied by the
controller to be physically realizable. This is natural since for real
applications, the controller should represent a physical system.
The constraint (\ref{eq:cntrl-pr-1}) is a non-convex, non-linear
equality constraint on the controller matrices $A_K$, $B_{K1}$,
$B_{K2}$, $B_{K3}$ and $C_K$ that presents a formidable challenge
in the controller design.

At present we do not know if there exists an exact or analytical
solution to Problem \ref{pb:lqg-main}. Moreover, in our experience
this non-convex problem is difficult to solve numerically using a
general purpose optimizer such as the `fmincon' routine in the
Matlab Optimization Toolbox; see \cite{Math07}. In our investigation,
it has been more fruitful to consider the relaxed problem of
finding a controller that achieves the cost bound
$J_{\infty}<\gamma$ for some pre-specified bound $\gamma>0$, and
to reformulate this problem into a rank constrained LMI
feasibility problem. As in $H^{\infty}$ synthesis, a solution to
Problem \ref{pb:lqg-main} can, in principle, be found iteratively
by employing a bisection method, or a variant thereof, on the
bound $\gamma$. Therefore, we focus our attention instead on the
following problem:

\begin{problem}\label{pb:lqg-form}
Given a fixed choice of $\Theta_K$ and cost bound parameter
$\gamma>0$, find controller matrices $A_K$, $B_{K1}$, $B_{K2}$,
$B_{K3}$ and $C_K$ such that the following conditions hold.

\begin{enumerate}
\item[F1.] There exists a symmetric matrix $P>0$ satisfying
(\ref{LQG-lyapunov}).

\item[F2.] $J_{\infty}=\tr(\mathcal C P \mathcal C^{T} )< \gamma$.

\item[F3.] The physical realizability constraints
(\ref{eq:cntrl-pr-1}) and (\ref{eq:cntrl-pr-2}) are satisfied.

\end{enumerate}
\end{problem}

\section{Reformulation of the quantum LQG problem into a rank
constrained LMI problem} \label{sec:lmi-rank} We shall now discuss
how to transform Problem \ref{pb:lqg-form} into a rank constrained
LMI problem which is amenable to numerical methods. To best
illustrate the idea, we opt to restrict our attention to the case
where $\Theta_K$ is canonical. Moreover, to  facilitate easy and
explicit exposition of the matrix lifting and linearization
technique, we shall take for a ``canonical'' example, a plant and
controller of order $n$ (recall that in our setup we are looking
for a controller which is of the same order as the plant) with
$n_y=n_u=n$ and $B_{K1}$, $B_{K2}$, $B_{K3}$, $C_K$ all of
dimension $n \times n$. Nonetheless, the matrix lifting principle
described for this canonical case can in principle be adapted to
more general scenarios and for the case in which $\Theta_K$ is
degenerate canonical, that is, the case where the controller has
both quantum and classical degrees of freedom. However, the
lifting is \emph{not} necessarily unique and is too complicated to
describe in a general form. More importantly, for efficiency the
choice of suitable lifting variables should in any case be
considered on a case by case basis to exploit any existing
structure in a particular problem.

Consider a $n$-th order plant (\ref{sys}) with $n_y=n_u=n$ and a
$n$-th order controller (\ref{controller}) with
$n_{w_{K1}}=n_{w_{K2}}=n$ (hence $B_{K1},B_{K2} \in \mathbb R^{n
\times n}$). Then $P$ will be a symmetric matrix of dimension $2n
\times 2n$. The first step is to transform the constraints
(\ref{LQG-lyapunov}) and $J_{\infty}<\gamma$ into an LMI
constraint. To do this we exploit a non-linear change of variables
given in \cite[Eq.(35)]{SGC97}, but to do this we first need to
suitably redefine our plant and controller equations while leaving
the closed-loop equations unaltered. To this end, let us redefine
our plant as:
\begin{eqnarray}
dx(t) &=& Ax(t)dt + B \beta_u(t) + B'_{w'} dw'(t); \quad x(0)=x; \nonumber \\
dy'(t) &=& C'x(t)dt +D'_{w'}dw'(t); \nonumber \\
z(t) &=& C_z x(t)+ D_z \beta_u(t), \label{sys-a}
\end{eqnarray}
with $w'=[\begin{array}{ccc} w\trp & w_{K1}\trp & w_{K2}\trp
\end{array}]\trp$, $B'_{w'}=[\begin{array}{ccc} B_w & B & 0_{n \times
n}
\end{array}]$,
$C'=[\begin{array}{ccc} 0_{n \times n} & 0_{n \times n} & C\trp
\end{array}]\trp$ and $$D'_{w'}=\left[\begin{array}{ccc} 0_{n \times n_w } &
I_{n \times n} & 0_{n \times n} \\ 0_{n \times n_w} & 0_{n \times
n_w} & I_{n \times n}\\ D_{w} & 0_{n \times n} & 0_{n \times
n}\end{array}\right].$$ Here $y'$ is the output equation for the
modified plant that now includes the quantum noise $w_{K2}$ that
enters in the controller, but not in the original plant. In this
way all noises can now be thought of as coming from the modified
plant, as in the setup of standard classical LQG problems. Then, we
also redefine our controller equations as:
\begin{eqnarray}
d\xi(t) &=& A_K\xi(t)dt + B_K dy'(t); \nonumber \\
\beta_u(t) &=& C_K\xi(t), \label{controller-a}
\end{eqnarray}
with $B_{K}=[\begin{array}{ccc} B_{K1} &  B_{K2} & B_{K3}
\end{array}]$. It is easily seen that interconnecting (\ref{sys-a}) and
(\ref{controller-a}) gives the closed-loop equation
(\ref{eq:closed-loop}). Now we are in the setup of \cite{SGC97}
with $D_K=0$ in \cite[Eq.(2)]{SGC97}.

We now follow \cite{SGC97}  by introducing auxiliary
variables $N$, $M$, $\mathbf X$, $\mathbf Y$, $Q \in \mathbb R^{n \times n}$,
with $\mathbf X,\mathbf Y,Q$ symmetric, and applying the following
non-linear change of variables (see \cite[Section IV-B]{SGC97}
with $\mathbf{\hat D}=D_K=0$):
\begin{eqnarray}
\mathbf A &=& NA_KM\trp + NB_K C'\mathbf X + \mathbf Y B C_K
M\trp+ Y A \mathbf X; \label{eq:aux-var-1}\\
\mathbf B&=&NB_K; \label{eq:aux-var-2}\\
\mathbf C&=& C_K M\trp. \label{eq:aux-var-3}
\end{eqnarray}
Then, the constraints (\ref{LQG-lyapunov}) and $J_{\infty}<\gamma$ can
be rewritten as the LMI constraint \cite[Eq.(14)]{SGC97}:
\begin{eqnarray}
&& \left[\begin{array}{cc} A\mathbf X + \mathbf X A\trp+ B\mathbf
C +(B\mathbf C)\trp & \mathbf A\trp + A \\
\mathbf A + A\trp & A\trp \mathbf Y + \mathbf Y A + \mathbf B C' +
(\mathbf B C')\trp \\
(B'_{w'})\trp & (\mathbf YB'_{w'}+ \mathbf B D'_{w'})\trp
\end{array}\right. \nonumber \\
&& \quad \left. \begin{array}{c} B'_{w'}  \\ \mathbf YB'_{w'}+
\mathbf B D'_{w'} \\ -I\end{array}\right]<0; \label{eq:LMI-1} \\
&&\left[\begin{array}{ccc} \mathbf X & I &
(C_z\mathbf X + D_z
\mathbf C)\trp \\
I &  \mathbf Y & C_z\trp \\
C_z \mathbf X + D_z \mathbf C & C_z & Q
\end{array}\right] > 0; \label{eq:LMI-2}\\ %
&&\tr(Q)<\gamma. \label{eq:LMI-3}
\end{eqnarray}
Since the controller is of the same order as the plant, the
matrices $N$ and $M$ can be freely chosen to be any pair of
(invertible) square matrices satisfying $MN\trp=I-\mathbf X
\mathbf Y$.

Once matrices $\mathbf A, \mathbf B, \mathbf C, \mathbf X, \mathbf
Y, Q$ satisfying the LMIs (\ref{eq:LMI-1})-(\ref{eq:LMI-3}) and
matrices $N$ and $M$ satisfying the conditions of the last
paragraph have been found, the original controller matrices
$A_K,B_K,C_K$ can be reconstructed as \cite[Eq.(40)]{SGC97}:
\begin{eqnarray}
C_K&=&\mathbf{C}M^{-T}; \label{eq:inv-LMI-1}\\
B_K&=&N^{-1}\mathbf B; \label{eq:inv-LMI-2}\\
A_K&=&N^{-1}(\mathbf A-NB_K C'\mathbf X -\mathbf Y B C_K M\trp
-\mathbf Y A \mathbf X)M^{-T}. \label{eq:inv-LMI-3}
\end{eqnarray}

Multiplying the left and right hand sides of (\ref{eq:cntrl-pr-1})
with $N$ and $N\trp$, respectively, and introducing new variables
$\breve N=N\Theta_K$ (keep in mind here that $\Theta_K$ is a fixed
matrix), $\breve A_K=NA_K$ and $\breve B_{Ki}=NB_{Ki}$,
$i=1,2,3$, (\ref{eq:cntrl-pr-1}) and (\ref{eq:cntrl-pr-2}) can be
expressed as:
\begin{eqnarray}
&&(-\mathbf AM^{-T}+(\breve B_{K3} C + \mathbf YA)\mathbf XM^{-T}
+\mathbf Y B
C_K)\breve N\trp \nonumber\\
&& \quad + \breve N (\mathbf AM^{-T}- (\breve B_{K3} C + \mathbf Y
A)\mathbf XM^{-T} -\mathbf Y B C_K)\trp+\nonumber\\
&&\quad \sum_{i=1}^{3}\breve
B_{Ki}\,\diag_{n/2}(J)\,\breve B_{Ki}\trp=0;\label{eq:cntrl-pr-12}\\
&&\breve B_{K1}=\breve N C_K\trp \,\diag_{n/2}(J).
\label{eq:cntrl-pr-22}
\end{eqnarray}
Conversely, if $A_K$, $B_{K1}$, $B_{K2}$, $B_{K3}$, $C_K$ solve
Problem \ref{pb:lqg-form} and $P$ is the solution of
(\ref{LQG-lyapunov}), then the LMIs
(\ref{eq:LMI-1})-(\ref{eq:LMI-3}) are all satisfied for some pair
of square matrices $M$ and $N$ satisfying $MN\trp=I-\mathbf X
\mathbf Y$; see \cite[Section IV-B]{SGC97}. Furthermore, since the
solution is physically realizable, (\ref{eq:cntrl-pr-12}) and
(\ref{eq:cntrl-pr-22}) are also satisfied. We summarize the
preceding discussion in the following theorem:

\begin{theorem}
\label{thm:trans-iff} Under the assumptions of this section,
Problem \ref{pb:lqg-form} has a solution for a given $\gamma>0$ if
and only if there exist matrices  $\mathbf A$, $\breve B_{K1}$,
$\breve B_{K2}$, $\breve B_{K3}$, $\mathbf C$, $\mathbf{X}$,
$\mathbf{Y}$, $\breve N$, $M$, $N$, $C_K $ satisfying the LMIs
(\ref{eq:LMI-1})-(\ref{eq:LMI-3}) (with $\mathbf
B=[\begin{array}{ccc} \breve B_{K1} & \breve B_{K2} & \breve
B_{K3} \end{array}]$) and the constraints
(\ref{eq:cntrl-pr-12})-(\ref{eq:cntrl-pr-22}) , $\breve
N=N\Theta_K$, $N M\trp=I-\mathbf Y \mathbf X$ and $\mathbf C=C_K
M\trp$.
\end{theorem}

 Note that (\ref{eq:cntrl-pr-12}) and
(\ref{eq:cntrl-pr-22}) are \emph{polynomial matrix equality
constraints} in the parameters $(\mathbf A,\breve B_{K1},\breve
B_{K2},\breve B_{K3},\mathbf C,\mathbf{X},\mathbf{Y},\breve
N,M^{-T})$. By this we mean that they are equality constraints in
a matrix-valued multivariate polynomial with matrix-valued
variables. If we take as decision variables the elements of the
parameters (for symmetric variables such as $\mathbf X$ we need
only take the upper triangular elements) then this becomes a
collection of scalar multivariate polynomial equalities. It is
well known that by introducing additional variables, called
\emph{lifting variables}, and some auxiliary equality constraints,
general polynomial equality and inequality constraints can be
``linearized'' and transformed into linear equality and inequality
constraints in some symmetric positive semidefinite matrix $X$,
plus a rank one constraint: ${\rm rank}(X)=1$ \cite{NN94,BV97}.
However, converting polynomial matrix constraints into a
collection of scalar polynomial constraints may not be desirable
as the resulting scalar polynomials could be of orders much higher
than the order of the original matrix polynomial. Subsequently,
when there are many decision variables the resulting scalar
polynomial program may easily become too large to handle
numerically. This is true in our present case, if we are to
\emph{scalarize} (\ref{eq:cntrl-pr-12}) then we would end up with
34 decision variables and constraints in polynomials of order 4, a
substantially large problem. Therefore we would like to keep the
matrix structure of our problem and try to find suitable matrix
lifting variables instead. The idea is similar to the scalar
version, but we need to take care of the fact that, unlike
scalars, matrices do not in general commute with one another.

We now proceed to linearize (\ref{eq:cntrl-pr-12}) and
(\ref{eq:cntrl-pr-22}) by introducing appropriate matrix lifting
variables and the associated equality constraints, and
transforming them into an LMI with a rank $n$ constraint. However,
in what will follow we set $M=I_{n\times n}$ and $N=I-\mathbf Y
\mathbf X$; this removes one free matrix variable, namely
$M^{-T}$, and \emph{reduces the complexity of the problem}. The 14
matrix lifting variables $W_1,W_2,\ldots,W_{14} \in \mathbb R^{n
\times n}$ are as follows: $W_i=\breve B_{Ki}J$, $i=1,2,3$,
$W_4=\mathbf Y B$, $W_5=\breve B_{K3} C+\mathbf Y A$, $W_6=\breve
N \mathbf C\trp$, $W_7=\breve N \mathbf X$, $W_8=\mathbf A \breve
N\trp$, $W_9=\mathbf Y \mathbf X$, $W_{10}=W_4 W_6\trp$,
$W_{11}=W_5 W_7\trp$, $W_{12}=W_1 \breve B_{K1}\trp$, $W_{13}=W_2
\breve B_{K2}\trp$ and $W_{14}=W_3 \breve B_{K3}\trp$. Now, let
$Z$ be a $23n \times 23n$ symmetric matrix,
$$\mathbf Z_{i,j}=[Z_{kl}]_{k=in+1,(i+1)n,l=jn+1,(j+1)n},$$
$$x=(x_1,\ldots,x_8)=(1,2,\ldots,8),$$ and
$$v=(v_1,\ldots,v_{14})=(9,10,\ldots,22).$$ We require that $Z$ satisfy the constraints:
\begin{eqnarray}
\left. \begin{array}{cc} Z  \geq 0 & \mathbf Z_{v_6,1}-\mathbf Z_{x_8,x_5}=0 \\
\mathbf
Z_{0,0}-I_{n \times n}=0 & \mathbf Z_{v_7,1}-\mathbf Z_{x_8,x_6}=0 \\
\mathbf Z_{1,x_6}-\mathbf
Z_{x_6,1}=0 & \mathbf Z_{v_8,1}-\mathbf Z_{x_1,x_8}=0 \\
\mathbf Z_{1,x_7}-\mathbf
Z_{x_7,1}=0 & \mathbf Z_{v_9,1}-\mathbf Z_{x_7,x_6}=0  \\
\mathbf Z_{v_1,1}-\mathbf Z_{x_2,1}\,\diag_{n/2}(J)=0 & \mathbf Z_{v_{10},1}-\mathbf Z_{v_4,v_6}=0  \\
\mathbf Z_{v_2,1}-\mathbf Z_{x_3,1}\,\diag_{n/2}(J)=0 & \mathbf Z_{v_{11},1}-\mathbf Z_{v_5,v_7}=0  \\
\mathbf Z_{v_3,1}-\mathbf Z_{x_4,1}\,\diag_{n/2}(J)=0 & \mathbf Z_{v_{12},1}-\mathbf Z_{v_1,x_2}=0  \\
\mathbf Z_{v_4,1}-\mathbf Z_{x_7,1}B=0 & \mathbf Z_{v_{13},1}-\mathbf Z_{v_2,x_3}=0 \\
\mathbf Z_{v_5,1}-\mathbf Z_{x_4,1}C-\mathbf Z_{x_7,1}A=0 & \mathbf Z_{v_{14},1}-\mathbf Z_{v_3,x_4}=0\\
\mathbf Z_{x_{8},1}-\Theta_K+\mathbf Z_{v_9,1}\Theta_K=0. &
\end{array} \right\} \label{eq:M-lift-cons}
\end{eqnarray}
Here, terms of the form $\mathbf{Z}_{a,b}$ with $a,b \in
\{x_1,\ldots,x_8\} \cup \{v_1,\ldots,v_{14}\}$ should be
identified with $\mathbf{Z}_{a,1}(\mathbf{Z}_{b,1})^T$. The LMI
constraints (\ref{eq:LMI-1})-(\ref{eq:LMI-3}) can be expressed in
terms of $Z$ by replacing $\mathbf A$, $\mathbf B$, $\mathbf C$,
$\mathbf X$, $\mathbf Y$ respectively with \\$\mathbf
Z_{x_1,1},\left[\begin{array}{ccc} \mathbf Z_{x_2,1} & \mathbf
Z_{x_3,1} & \mathbf Z_{x_4,1} \end{array} \right],\mathbf
Z_{x_5,1}, \mathbf Z_{x_6,1}, \mathbf Z_{x_7,1}$, while the
physical realizability constraints (\ref{eq:cntrl-pr-12}) and
(\ref{eq:cntrl-pr-22}) become the following pair of linear
equality constraints:
\begin{eqnarray}
&&-\mathbf Z_{v_8,1}+\mathbf Z_{v_8,1}\trp+\mathbf
Z_{v_{11},1}-\mathbf Z_{v_{11},1}\trp+\mathbf Z_{v_{10},1}-
\mathbf Z_{v_{10},1}\trp
\nonumber \\&&\qquad
+\mathbf Z_{v_{12},1}+\mathbf Z_{v_{13},1}+\mathbf Z_{v_{14},1}=0; \nonumber \\
&&\mathbf Z_{x_2,1}-\mathbf Z_{v_6,1}\diag_{n/2}(J)=0.
\label{eq:linearized-cons}
\end{eqnarray}
Finally, we also require that $Z$ satisfy a rank $n$ constraint:
\begin{equation}
{\rm rank}(Z)\leq n. \label{eq:rank-M}
\end{equation}

To understand the above rank constrained LMI and its relation to
our original constraints, suppose that there is a $Z$ satisfying
(\ref{eq:M-lift-cons})-(\ref{eq:rank-M}) and the LMI constraints
(expressed in terms of block elements of $Z$). Then, since $Z \geq
0$ and is of rank at most $n$, we may factorize it as $Z=VV\trp$,
where $V \in \mathbb R^{23n \times n}$ and satisfies
$[V_{ij}]_{i,j=1,n}=I_{n \times n}$, and by (\ref{eq:M-lift-cons})
we recover $\mathbf A$, $\breve B_{Ki}$ $(i=1,2,3)$, $\mathbf C,
\mathbf X, \mathbf Y, \breve N$ respectively as $\mathbf
Z_{x_1,1},\ldots,\mathbf Z_{x_8,1}$, and also recover $W_i=\mathbf
Z_{v_i,1}$, $i=1,\ldots,14$. Then $N=\breve N\Theta_K^{-1}$ and
$\mathbf B_{Ki}=N^{-1} \breve B_{Ki}$ $(i=1,2,3)$. The controller
matrices $A_K,B_K,C_K$ are given by
(\ref{eq:inv-LMI-1})-(\ref{eq:inv-LMI-3}) and by construction they
will satisfy (\ref{eq:LMI-1})-(\ref{eq:LMI-3}),
(\ref{eq:cntrl-pr-1}) and (\ref{eq:cntrl-pr-2}). Thus, we obtain a
solution to Problem \ref{pb:lqg-form}.

It should be noted that due to the simplifying assumptions $M=I$
and $N=I-\mathbf Y \mathbf X$, solvability of the rank constrained
LMI problem formulated above is only \emph{sufficient} for
solvability of Problem \ref{pb:lqg-form}. In fact, in Theorem
\ref{thm:trans-iff} it is only required that $M$ and $N$ satisfy
$NM\trp=I-\mathbf Y \mathbf X$. Nonetheless, it is not difficult
to see that by (i) introducing additional variables $M$, $N$, $C_K
\equiv \mathbf C M^{-T}$, $\breve{\mathbf A} \equiv \mathbf A
M^{-T}$, $\breve{\mathbf X} \equiv \mathbf X M^{-T}$, (ii)
additional lifting variables $W_{15}=NM\trp$,
$W_{16}=\breve{\mathbf A} M\trp$, $W_{17}= \breve{\mathbf X}
M\trp$, $W_{18}=C_K M\trp $ and associated constraints
$W_{15}-I+W_9=0$, $\mathbf A-W_{16}=0$, $\mathbf X-W_{17}=0$,
$\breve N - N\Theta_K=0$, $\mathbf C-W_{18}=0$ (iii) redefining
$W_6=\breve N C_K\trp$, $W_7=\breve N \breve{\mathbf X}\trp$ and
$W_8=\breve{\mathbf A} \breve N\trp$, and (iv) enlarging and
redefining $Z$ as well as the set of constraints
(\ref{eq:M-lift-cons}) and (\ref{eq:linearized-cons}) accordingly,
these simplifying assumptions can be removed to make the resulting
rank constrained LMI problem also \emph{necessary} for solvability
of Problem \ref{pb:lqg-form}, but at the expense of having to
solve a larger problem. The simplification we have proposed here,
however, may be especially useful for reducing the complexity of
problems with a plant dimension larger than 2. We conclude this
section with the following  remark.

\begin{remark}
In the formulation of this section, it is not actually essential
to fix $\Theta_K$ to be $\diag_{n/2}(J)$. Instead, it may also be
fixed to be $\Theta_K^S=S\diag_{n/2}(J)S\trp$ for any real
invertible matrix $S$. Indeed, if $A_K^S$, $B_{Ki}^S \,
(i=1,2,3)$, $C_K^S$ solves Problem \ref{pb:lqg-form} for
$\Theta_K=\Theta_K^S$ then $A_K=S^{-1} A_K^S S$,
$B_{Ki}=S^{-1}B_{Ki}^S, \,(i=1,2,3)$, $C_K=C_K^S S$ solves Problem
\ref{pb:lqg-form} for $\Theta_K=\diag_{n/2}(J)$; for details, see
a related discussion in Section \ref{sec:extension} where
$\Theta_K$ is allowed to be a free variable. This added
flexibility will be valuable for numerical attempts at solving
Problem \ref{pb:lqg-form}.
\end{remark}

\section{Numerical solution of the rank constrained LMI problem}
\label{sec:numerics} We have seen in the preceding section that
our problem is essentially a polynomial matrix programming (to be
precise, feasibility) problem (since LMIs can themselves be viewed
as polynomial matrix inequalities) and that the latter can be
converted to a rank constrained problem. It is well-known that
many important practical control problems can be formulated as
polynomial programming problems, including reduced order robust
controller design, static output feedback and gain scheduling (see
\cite{HL06} and the references therein). They are non-convex and
non-linear problems that are, in general, difficult to solve. In
fact, some of these problems are known to be NP-hard
\cite{Shor90,Lass01}.

If one tries to directly attack the (scalar or matrix) polynomial
programming problem, then a specialized method for solving them is
to employ LMI relaxations techniques based on the theory of
moments and the dual theory of sum of squares (SOS) polynomials; see
\cite{Lass01,Koji03,HS04,HL06}. Under appropriate conditions,
relaxation methods can be guaranteed to converge as the order of
relaxation is increased and it can be checked whether a global
optima may have been obtained at a particular relaxation. Despite
its attractive features, the size of the relaxed LMI problem grows
very quickly with the number of decision variables, the degree of
polynomials involved and the order of relaxation, making the
method impractical for problems with many decision variables.

On the other hand, if the problem is converted to a rank
constrained LMI problem then there are \emph{iterative} algorithms
in the literature that try to directly search for a feasible point
satisfying the set of LMIs and the rank constraint, mostly based
on the idea of alternating projections (see \cite{OHM06} and the
references therein). The main drawback of these algorithms is that
they are difficult to analyze and are not in general guaranteed to
converge from arbitrary starting points, even if a solution
exists. However, since there are no relaxations involved that
increase the size of the problem to be solved, they can be more
attractive for solving medium and larger size polynomial
programming problems. This makes them more suitable for our
current problem, which can be considered to be of a substantial
size (recall that if it is converted to a scalar polynomial
programming problem then there would be 34 decision variables and
involve multivariate polynomials of up to degree 4).

To solve the rank constrained  LMI problem formulated in the last
section, we shall use an algorithm by \cite{OHM06} that has been
implemented in the freely available Matlab toolbox LMIRank (see
\cite{LMIRank}) and can be called via the Yalmip optimization
prototyping environment; see \cite{YALMIP}. This algorithm is also
based on alternating projections but, unlike previous alternating
projections algorithms, has a built-in Newton step that has the
potential to accelerate convergence. In Section
\ref{sec:examples}, we will use this method to numerically solve
an example coherent LQG control problem.

As mentioned earlier, solvers for rank constrained LMI problems
are not guaranteed to converge from arbitrary starting points.
Therefore, it is important to have a heuristic method for choosing
starting points for these algorithms. For a given $\gamma>0$, to
obtain a starting point for the LMIRank solver we suggest to first
solve the LMIs (\ref{eq:LMI-1})-(\ref{eq:LMI-3}) to obtain
$\mathbf A,\mathbf B, \mathbf C, \mathbf X, \mathbf Y, Q$. Then
set $M=I_{n \times n}$ and $N=I-\mathbf Y \mathbf X$ and compute
$\breve B_{K1}$, $\breve B_{K2}$, $\breve B_{K3}$, $\breve N$ and
the matrix lifting variables $W_1,\ldots,W_{14}$ according to the
definitions given in Section \ref{sec:lmi-rank}. Let
\begin{eqnarray*}
V_0&=&[\begin{array}{cccccccccccc} I_{n \times n} & \mathbf A\trp
& \breve B_{K1}\trp & \breve B_{K2}\trp & \breve B_{K3}\trp &
\mathbf{C}\trp & \mathbf{X}\trp & \mathbf{Y}\trp & \breve N\trp &
W_1\trp & \ldots & W_{14}\trp
\end{array}]\trp.
\end{eqnarray*}
Then we set $Z=V_0 V_0\trp$ as a heuristic starting point.

\section{An extension of the  numerical procedure}
\label{sec:extension} In this section, we discuss an important
extension of the methodology that has been developed in preceding
sections of the paper.

In Problem \ref{pb:lqg-form}, the commutation matrix $\Theta_K$ is
a fixed but freely specified parameter that determines the type of
controller sought. However, since $A_K$, $B_K$, $C_K$ are all
allowed to be free, it is not essential to restrict $\Theta_K$ to
be canonical or degenerate canonical. A similarity transformation
on $A_K$ essentially corresponds to a (matrix rescaling) of
$\Theta_K$ (i.e. pre and post-multiplication by a non-singular
matrix). Therefore it is possible to also allow $\Theta_K$ to be a
free variable that is only restricted to be \emph{real} and
\emph{skew symmetric}. In this setting, we do not \emph{a priori}
determine the type of controller that is sought (a classical,
quantum, or mixed classical-quantum controller). Instead, we also
\emph{seek} an optimal type of controller. This leads to a
modification of Problem \ref{pb:lqg-form} where the former
restriction on $\Theta_K$ is removed.

\begin{problem}\label{pb:lqg-form-2}
Given a cost bound parameter $\gamma>0$, find real matrices $A_K$,
$B_{K1}$, $B_{K2}$, $B_{K3}$, $C_K$ and a skew symmetric real
matrix $\Theta_K$ such that:

\begin{enumerate}
\item There exists a symmetric matrix $P>0$ satisfying
(\ref{LQG-lyapunov}).

\item $J_{\infty}=\tr(\mathcal C P \mathcal C^{T} )< \gamma$.

\item The resulting controller is physically realizable. That is,
$A_K$, $B_{K1}$, $B_{K2}$, $B_{K3}$, $C_K$ and $\Theta_K$ satisfy
(\ref{eq:cntrl-pr-1}) and (\ref{eq:cntrl-pr-2})
\end{enumerate}
\end{problem}

The following lemma is then immediate and shows that solving
Problem \ref{pb:lqg-form-2} also solves Problem \ref{pb:lqg-form}.

\begin{lemma}
\label{lem:pb2-to-pb1}Suppose that the matrices $\hat A_K$, $\hat
B_{K1}$, $\hat B_{K2}$, $\hat B_{K3}$, $\hat C_K$ and $\hat
\Theta_K$ solve Problem \ref{pb:lqg-form-2}, and $\hat \Theta_K=S
Z S\trp$ for some canonical or degenerate canonical matrix $Z$ and
some real invertible matrix $S$. Then the matrices
\begin{eqnarray}
A_K=S^{-1}\hat A_K S;\quad B_{Ki}=S^{-1} \hat B_{Ki} \quad
i=1,2,3; \quad C_K=\hat C_K S, \label{eq:aux-sol-2a}
\end{eqnarray}
solve Problem \ref{pb:lqg-form} for $\Theta_K=Z$.
\end{lemma}

\begin{proof}
Since $\hat \Theta_K$ is real skew symmetric, we can find an
invertible matrix $S$ such that $\hat \Theta_K=S Z S\trp$ for some
matrix $Z$ which is either canonical or degenerate canonical. Now,
we have that
\begin{eqnarray}
&&\hat A_K \hat \Theta_K + \hat \Theta_K \hat A_K\trp + \hat
B_{K1}\,\diag_{n_1/2}(J)\, \hat B_{K1}\trp +
\hat B_{K2}\,\diag_{n_2/2}(J)\, \hat B_{K2}\trp + \nonumber\\
&&\quad \hat B_{K3}\,\diag_{n_3/2}(J)\,\hat B_{K3}\trp=0; \label{eq:aux-sol-2b}\\
&&\hat B_{K1}=\hat \Theta_K \hat C_K\trp \diag_{n_u/2}(J)
\label{eq:aux-sol-2c}.
\end{eqnarray}
After substitution of (\ref{eq:aux-sol-2a}) and $\hat \Theta_K=S Z
S\trp$ into (\ref{eq:aux-sol-2b}) and (\ref{eq:aux-sol-2c}) and
some algebraic manipulations, it is easily obtained that $A_K$,
$B_K$ and $C_K$ satisfy
\begin{eqnarray}
&& A_K Z + Z A_K\trp +  B_{K1}\,\diag_{n_1/2}(J)\, B_{K1}\trp +
 B_{K2} \,\diag_{n_2/2}(J)\, B_{K2}\trp + \nonumber\\
&&\quad B_{K3}\,\diag_{n_3/2}(J)\, B_{K3}\trp=0; \label{eq:aux-sol-2d}\\
&& B_{K1}=Z C_K\trp \diag_{n_u/2}(J) \label{eq:aux-sol-2e}.
\end{eqnarray}
Moreover, the LQG cost is invariant under a similarity
transformation of the controller state space matrices. Hence,
these matrices solve Problem \ref{pb:lqg-form} with $\Theta_K=Z$,
as claimed. \hfill $\Box$
\end{proof}

The following corollary is then obvious.

\begin{corollary}
\label{cor:std-lqg-to-pb1} Let $\Theta_K$ be given (canonical or
degenerate canonical) and suppose that 
$A_K$, $B_{K1}$, $B_{K2}$, $B_{K3}$, $C_K$ solve the standard LQG problem for
a given $\gamma>0$ (i.e. (\ref{LQG-lyapunov}) and
$J_{\infty}<\gamma$ are satisfied), and there exists a real skew
symmetric $Z$ satisfying (\ref{eq:aux-sol-2d}) and
(\ref{eq:aux-sol-2e}). Moreover, suppose there exists a real
invertible matrix $S$ such that $SZS\trp=\Theta_K$. Then the
matrices $\hat A_K$, $\hat B_{K1}$, $\hat B_{K2}$, $\hat B_{K3}$,
$\hat C_K$ given by:

\begin{eqnarray}
\hat A_K=S A_K S^{-1}; \quad\hat B_{Ki}=S B_{Ki} \quad i=1,2,3;
\quad\hat C_K=C_K S^{-1}, \label{eq:aux-sol-2f}
\end{eqnarray}
solve Problem \ref{pb:lqg-form}.
\end{corollary}

The preceding corollary says that a solution $A_K,B_K,C_K$ to the
standard LQG problem will also solve the quantum LQG problem
(Problem \ref{pb:lqg-form}) \emph{if and only if} a matrix $Z$ can
be found satisfying (\ref{eq:aux-sol-2d}) and
(\ref{eq:aux-sol-2e}) and there exists a real invertible matrix
$S$ such that $SZS\trp=\Theta_K$.

It is also possible to treat Problem \ref{pb:lqg-form-2} using the
rank constrained LMI procedure of Section \ref{sec:lmi-rank}. To
do this, again under the simplifying assumptions $M=I$ and
$N=I-\mathbf Y \mathbf X$ (but which, as remarked earlier, can be
easily removed if desired), introduce the additional variable
$\Theta_K$ and substitute $A_K$, $B_{Ki}$ ($i=1,2,3$), $C_K$ with,
respectively, $\hat A_K$, $\hat B_{Ki}$ ($i=1,2,3$), $\hat C_K$
(cf. Lemma \ref{lem:pb2-to-pb1}). Then we redefine $\breve
N=N\Theta_K$, $x=(x_1,\ldots,x_{10})$ and $Z$ to be a real
symmetric matrix of dimension $25n \times 25n$, and replace
(\ref{eq:M-lift-cons}) with the following set of constraints:
\begin{eqnarray}
\left. \begin{array}{cc}
 Z  \geq 0 & \mathbf Z_{v_3,1}-\mathbf Z_{x_4,1}\,\diag_{n/2}(J)=0  \\
 \mathbf Z_{0,0}-I_{n \times n}=0 & \mathbf Z_{v_4,1}-\mathbf Z_{x_7,1}B=0 \\
 \mathbf Z_{1,x_6}-\mathbf Z_{x_6,1}=0 & \mathbf Z_{v_5,1}-\mathbf Z_{x_4,1}C-\mathbf Z_{x_7,1}A=0 \\
 \mathbf Z_{1,x_7}-\mathbf Z_{x_7,1}=0 & \mathbf Z_{v_6,1}-\mathbf Z_{x_8,x_5}=0  \\
 \mathbf Z_{v_1,1}-\mathbf Z_{x_2,1}\,\diag_{n/2}(J)=0 & \mathbf Z_{v_7,1}-\mathbf Z_{x_8,x_6}=0  \\
 \mathbf Z_{v_2,1}-\mathbf Z_{x_3,1}\,\diag_{n/2}(J)=0 & \mathbf Z_{v_8,1}-\mathbf Z_{x_1,x_8}=0 \\
 \mathbf Z_{x_{9},1}-I_{n \times n}+\mathbf Z_{v_9,1}=0 & \mathbf Z_{x_{10},1}+ \mathbf Z_{1,x_{10}}=0 \\
 \mathbf Z_{v_9,1}-\mathbf Z_{x_7,x_6}=0 &  \mathbf Z_{v_{13},1}-\mathbf Z_{v_2,x_3}=0\\
 \mathbf Z_{v_{10},1}-\mathbf Z_{v_4,v_6}=0 & \mathbf Z_{v_{14},1}-\mathbf Z_{v_3,x_4}=0  \\
 \mathbf Z_{v_{11},1}-\mathbf Z_{v_5,v_7}=0 & \mathbf Z_{x_{8},1}+ \mathbf Z_{x_9,x_{10}}=0 \\
\mathbf Z_{v_{12},1}-\mathbf Z_{v_1,x_2}=0. &
\end{array} \right\} \label{eq:M-lift-cons-2}
\end{eqnarray}
Here, $\mathbf Z_{x_{10},1}+ \mathbf Z_{1,x_{10}}=0$ accounts for
the requirement that $\Theta_K=-\Theta_K\trp$ while $\mathbf
Z_{x_{8},1}+ \mathbf Z_{x_9,x_{10}}$ accounts for the the
constraint $\breve N=N\Theta_K$. The remaining constraints of
Section 5, (\ref{eq:linearized-cons}) and ${\rm rank}(Z)\leq n$,
remain the same. Note that the variable $N$ is now also an {\em
independent} variable, whereas in Section \ref{sec:lmi-rank} $N$
was {\em linearly related} to $\breve N$ by the relation $N=\breve
N \Theta_K^{-1}$ for some fixed constant $\Theta_K$. However, the
introduction of additional free variables and constraints increase
the complexity of the problem and may have a bearing on the
convergence of the alternating projections algorithm. Therefore,
the solver for Problem \ref{pb:lqg-form-2} may be used to
\emph{complement}  the solver for Problem \ref{pb:lqg-form}, and
substitute for one another in case convergence fails for one of
them. Furthermore, it is also necessary to develop an additional
heuristic to determine a good initial guess $\Theta_K^{0}$ for
$\Theta_K$. These can be topics to be considered for further
research. Possible, albeit arbitrary, choices for $\Theta_K^{0}$
are $\Theta_K^{0}=0_{2 \times 2}$ or $\Theta_K^{0}=J$. Once
$\Theta_K^{0}$ has been chosen, we set $N=I-\mathbf Y \mathbf X$,
$\breve N^0=N\Theta_K^0$ and
\begin{eqnarray*}
V_0=[\begin{array}{cccccccccccccc} I & \mathbf A\trp & \breve
B_{K1}\trp & \breve B_{K2}\trp & \breve B_{K3}\trp &
\mathbf{C}\trp & \mathbf{X}\trp & \mathbf{Y}\trp & (\breve
N^0)\trp & N\trp & (\Theta_K^0)\trp & W_1\trp & \ldots &
W_{14}\trp
\end{array}]\trp,
\end{eqnarray*}
then the alternating projections algorithm can be executed by
setting $Z=V_0 V_0\trp$ as the starting point.

\begin{remark}
It should be noted that the extended procedure should not be
considered to supersede the procedure for strictly finding a
quantum controller. This is because, as discussed in the
Introduction, there can be circumstances where a quantum
controller is desirable while the extended procedure may not
return such a controller since $\Theta_K$ is allowed to be free.
\end{remark}

\section{Quantum LQG control design examples}
\label{sec:examples} In this section, we apply the transformation
and matrix lifting technique of Section \ref{sec:lmi-rank} to
compute a fully quantum LQG controller to asymptotically stabilize
a marginally stable fully quantum plant. We work in Matlab using
the Yalmip prototyping environment and a solution was computed
using LMIRank. The semidefinite program solver used for LMIRank is
SeDuMi Version 1.1 Release 3; see \cite{SeDuMi11}. Then for
comparison, we also compute a classical controller that  modulates
a light beam to drive the plant. All computations were performed
on Matlab running on an Apple Mac Pro workstation configured with
two 3GHz Quad-Core Intel Xeon processors, 8GB of memory and
232.89GB hard disk capacity.

The quantum plant to be controlled is a physically realizable (cf.
Section \ref{sec:phys-real}) fully quantum system with Hamiltonian
matrix $R$ and coupling matrix $\Lambda$ (see \cite{JNP06} for the corresponding
definitions) given by
$$R=\frac{1}{2}\left[\begin{array}{cc} \Delta & 0 \\
0 & \Delta \end{array}\right] \quad
 \Lambda=\left[\begin{array}{cc} \sqrt{\kappa_1} & 0 \\
\sqrt{\kappa_2} & 0 \\ \sqrt{\kappa_3} & 0 \end{array}\right].$$
Then,  its dynamics  are given by
\begin{eqnarray}
dx &=& \left[\begin{array}{cc} 0 & \Delta \\-\Delta & 0
\end{array}\right] x dt + \left[ \begin{array}{cc} 0 & 0 \\ 0 & -2\sqrt{k_1}
\end{array}\right]du +  \left[\begin{array}{cccc} 0 & 0 & 0 & 0
\\ 0 & -2\sqrt{k_2} & 0 & -2\sqrt{k_3} \end{array}\right]
\left[\begin{array}{c} dw_1 \\ dw_2 \end{array}\right],
\nonumber\\
dy&=&\left[\begin{array}{cc} 2\sqrt{k_2} & 0
\\ 0 & 0 \end{array}\right]xdt + dw_1 \label{eq:eg-plant-1}
\end{eqnarray}
with $\Delta=0.1$ and $k_1=k_2=k_3=10^{-2}$. This plant may be
thought of as representing the scenario of an atom trapped between
two mirrors of a three mirror cavity in the strong coupling limit
in which the cavity dynamics are adiabatically eliminated; see
\cite{GvH07,DJ99}. Note that by definition of the coupling matrix
$\Lambda$, the quantum noise fields couple only to the position
operator of the atom, which is also the typical setup sought in
various schemes for quantum non-demolition continuous measurement
of position. This particular choice of coupling results in a
marginally stable plant with $A$ having two mutually conjugate
eigenvalues on the imaginary axis.

\subsection{Quantum LQG controller design example I}
Let us try to asymptotically stabilize this system with another
quantum system as the LQG controller. To this end we set $z= x +
\xi$; i.e.,  $C_z=I_{2\times 2}=D_z$. Choosing $\gamma=5.75$ and
numerically solving Problem \ref{pb:lqg-form} following Sections
\ref{sec:lmi-rank} and \ref{sec:numerics}, yields the following
physically realizable controller after 1000 iterations of LMIRank
(with a running time of 2944.9 seconds):
\begin{eqnarray}
d\xi &=& \left[\begin{array}{cc} -2.3907  & 0.8420 \\
-5.5518 &
1.9380\end{array}\right] \xi dt + \left[ \begin{array}{cc} -0.3029 & 0.5042 \\
-0.6603 & 1.0819\end{array}\right]dw_{K1} + \nonumber \\
&&\quad
10^{-10}\left[ \begin{array}{cc} 0.0241 & -0.0471 \\
0.0576 & -0.1136 \end{array}\right]dw_{K2}+
\left[\begin{array}{cc} 3.3626 & 2.1470  \\ 7.6699 & 5.0302
\end{array}\right]dy  \nonumber \\
du(t)&=&\left[\begin{array}{cc} -1.0819 & 0.5042  \\
-0.6603 & 0.3029
\end{array}\right]xdt+dw_{K1} \label{eq:eg-cntrl-1}
\end{eqnarray}
that asymptotically stabilizes the closed-loop system. The actual
closed loop LQG cost achieved by this controller can be computed
to be $J_{\infty}=5.7382$.

Notice that elements of $B_{K2}$ are very small (of the order
$10^{-10}$). This is  a result that we would ideally like to have since $B_{K2}$
is the coefficient for the quantum noise $w_{K2}$ that {only
enters} in the controller and does {\em not} come from the plant.
This noise contributes towards the LQG cost, but since the aim is
to bound this cost, it is not  surprising that the
algorithm finds a controller for which $B_{K2}$ is effectively zero, in order
to remove the effect of the variance of $w_{K2}$ on the LQG cost.
Up to the numerical precision of Matlab, the above numerical
results give:
\begin{eqnarray}
&&A_K \Theta_K +\Theta_K A_K^T + \sum_{k=1}^{3} B_{Ki}JB_{Ki}^T=
10^{-13}\left[\begin{array}{cc} 0 & 0.2896 \\ -0.2896 &
0\end{array} \right] \label{eq:eq1-ccr-num},
\end{eqnarray}
while dropping the quadratic term containing $B_{K2}$ (i.e.,
setting $B_{K2}=0$) returns an identical numerical result on the
right hand side of (\ref{eq:eq1-ccr-num}). This indicates that the
contribution of $B_{K2}$ to (\ref{eq:eq1-ccr-num}) is less that
the numerical precision of Matlab and we may in this case simply
set $B_{K2}=0$ to obtain a simpler controller.

\subsection{Classical LQG controller designs}
After successfully obtaining a fully quantum controller, a natural
question that now arises is: Does this controller offer any
improvement over a classical controller that is driven by
continuous measurements (e.g., by homodyne detection; see
\cite{BR04}) of a quadrature of the plant output $y$? To answer
this question, suppose now that we perform continuous measurements
of one quadrature (in this case, the first element) of $y$. Thus,
we replace the output $y$ in (\ref{eq:eg-plant-1}) by another
output $y'$ (a classical signal) given by:
$$y'=[\begin{array}{cc} 2\sqrt{k_2} & 0 \end{array}]xdt+[\begin{array}{cc} 1 & 0 \end{array}]dw_1.$$
We seek a classical controller of the form:
\begin{eqnarray}
d\xi &=& A_K\xi dt + B_K dy'; \nonumber \\
\beta_u &=& C_K \xi \label{eq:c-cntrl}
\end{eqnarray}
whose output  modulates  a light beam to produce the control
signal $u$:
\begin{eqnarray}
du&=&\beta_u dt+ dw_{K1} \nonumber \\
&=&C_K \xi dt + dw_{K1}. \label{eq:c-cntrl-out}
\end{eqnarray}
The optimal controller matrices $A_K,B_K,C_K$ can now be found by
applying the standard LQG machinery (with $C_z$ and $D_z$ as
before) to the following modified plant (to account for the
presence of the noise $w_{K1}$ in the controller output $u$, see a
related discussion in Section \ref{sec:lmi-rank}) with $\beta_u$
being viewed as the ``control signal'':
\begin{eqnarray*}
dx &=& \left[\begin{array}{cc} 0 & \Delta \\-\Delta & 0
\end{array}\right] x dt + \left[ \begin{array}{cc} 0 & 0 \\ 0 & -2\sqrt{k_1}
\end{array}\right]\beta_u dt
\\&&\quad   +\left[\begin{array}{cccccc} 0 & 0 & 0 & 0 & 0 & 0
\\ 0 & -2\sqrt{k_1} & 0 & -2\sqrt{k_2} & 0 & -2\sqrt{k_3} \end{array}\right]
\left[\begin{array}{c} dw_{K1} \\ dw_1 \\ dw_2
\end{array}\right],
\nonumber\\
dy'&=&[\begin{array}{cc} 2\sqrt{k_2} & 0
\end{array}]x dt+[\begin{array}{cc} 1 & 0 \end{array}]dw_1.
\end{eqnarray*}
The optimal classical controller was then found to be:
\begin{eqnarray*}
d\xi &=&\left[\begin{array}{cc} -0.0658 & 0.1 \\ -0.1217 & -0.2 \end{array}\right]\xi dt + \left[\begin{array}{c} 0.3291 \\ 0.1083 \end{array}\right]dy'; \\
du &=& \left[\begin{array}{cc} -1 & 0 \\ 0 & 1
\end{array}\right]\xi dt + dw_{K1},
\end{eqnarray*}
while the optimal LQG cost is $J_{\infty}=4.8468$.

It turns out that cost obtained by this classical controller is
lower than the cost we had previously obtained with the quantum
LQG controller ($J_{\infty}=4.8468$ in the former compared to
$J_{\infty}=5.7382$ in the latter). Moreover, this may not the
best performance achievable by a classical linear controller as it
is also possible to perform a so-called {\em indirect measurement}
of $y$ by first mixing it with an additional vacuum noise via a
beam splitter \cite{BR04,GZ00}, followed by homodyne detection of
the real quadrature of one output of the beamsplitter and of the
imaginary quadrature of the other output, thus giving the
controller \emph{noisy} information about \emph{both} quadratures
of $y$. For examples of indirect measurement, see \cite[Sections
VII-C and VII-D]{JNP06}.

To describe an indirect measurement process, suppose that the
vacuum noise going into the beamsplitter is
$w_0=(w_{0,1},w_{0,2})$ (in quadrature notation) and the
beamsplitter divides the signals between the two ports according
to the ratio $\alpha:\beta$ where $\alpha^2+\beta^2=1$, $0\leq
\alpha,\beta \leq 1$. Then the classical output signal $y''$ of
the indirect measurement is given by:
$$dy''=\left[\begin{array}{cc} \alpha & 0 \\ 0 & -\beta
\end{array}\right]dy+ \left[\begin{array}{cc} \beta & 0 \\ 0 &
\alpha
\end{array}\right]dw_0.$$

Let us now  seek a classical controller of the form
(\ref{eq:c-cntrl}) and (\ref{eq:c-cntrl-out}) with $y'$ replaced
by $y''$ for various values of $\alpha$ between $0$ and $1$. Note
that measurement of $y'$ can be viewed as a special case of
indirect measurement with $\alpha=1$. A plot of the LQG cost
$J_{\infty}$ achieved by this controller versus the value of
$\alpha$ is shown in Figure~\ref{fg:indirect}.

\begin{figure}[h!]
\centering
\includegraphics[scale=0.4]{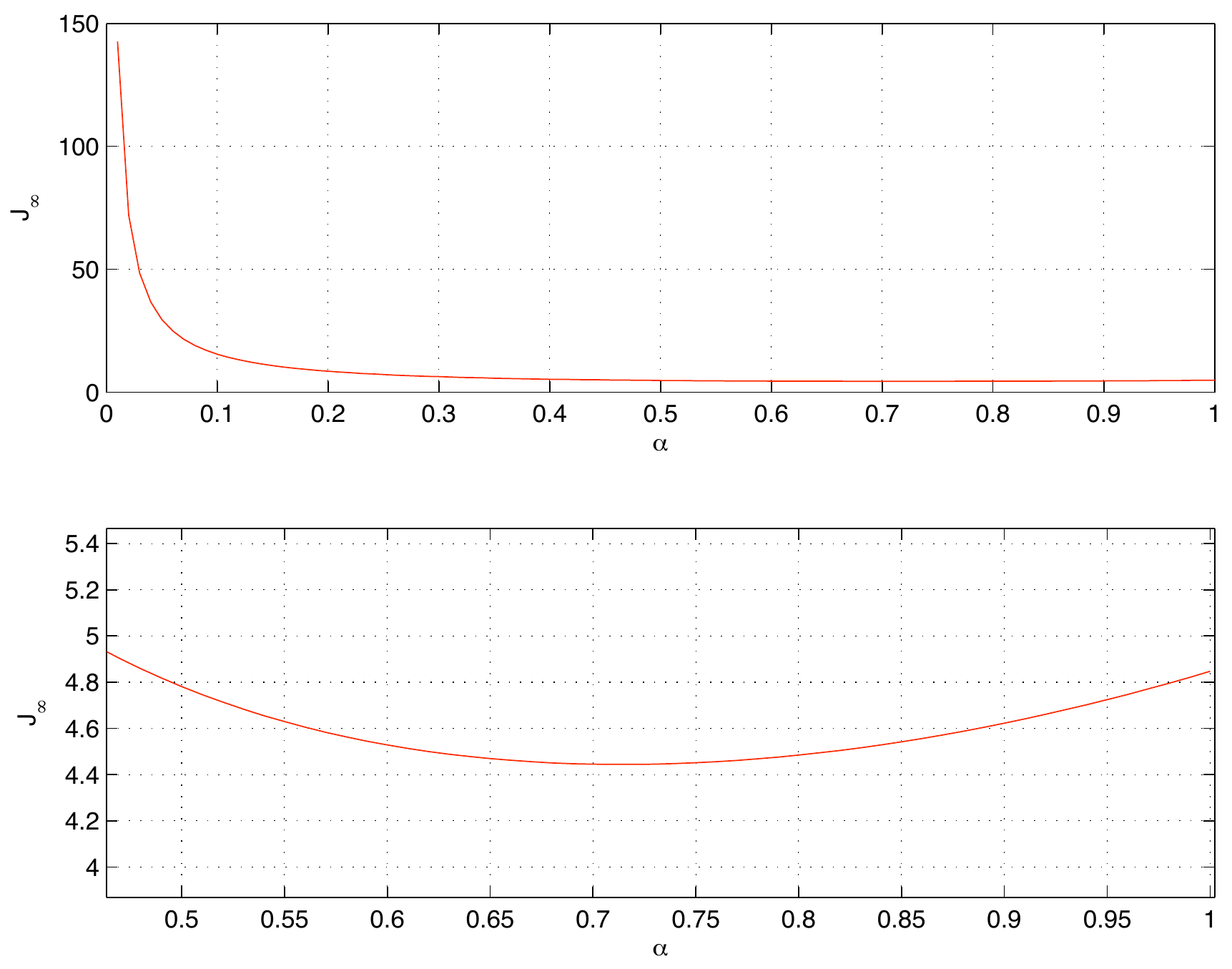}
\caption{Plot of $J_{\infty}$ vs. $\alpha$}%
\label{fg:indirect}%
\end{figure}

It can be seen that $J_{\infty}$ achieves its minimum around
$\alpha=0.715$ with a value of $J_{\infty} \approx 4.4440$. This
raises the question of whether there exists at all a fully quantum
controller that gives a cost lower the best achievable with
indirect measurements. We shall see in the next section that the
answer to this is affirmative by obtaining such a fully quantum
controller.

\subsection{Quantum LQG controller design example II}
We now try to find a fully quantum controller that achieves a cost
$J_{\infty}<4.444$. However, several attempts at numerically
solving Problem \ref{pb:lqg-form} for $\Theta_K=J$ by alternating
projections were unsuccessful with the algorithm failing to
converge. Thus, we turn our attention to numerically solving
Problem \ref{pb:lqg-form-2} according to the discussion in Section
\ref{sec:extension} with $\gamma=5$ and $\Theta_K^0=0_{2 \times
2}$. It turns out that the alternating projections algorithm
converges (with a running time of approximately 781.87 seconds)
and returns the following solution:
\begin{eqnarray*}
&&\Theta_K=\left[\begin{array}{cc} 0 & -0.1820 \\
0.1820 & 0 \end{array}\right] ;\quad \hat
A_K=\left[\begin{array}{cc} -0.2125 & 0.0666 \\ -0.3789 & 0.0257
\end{array}\right];\\
&& \quad \hat B_{K1}=\left[\begin{array}{cc} 0.0642 & -0.0547 \\
0.0480 & -0.2556
\end{array}\right]; \quad \hat B_{K2}=10^{-10}\left[\begin{array}{cc}-0.0468 &
0.0255 \\ -0.1160 & 0.0114
\end{array}\right];\\
&&\hat B_{K3}=\left[\begin{array}{cc} 0.3522 & -0.0215 \\ 0.4393 &
-0.0842
\end{array}\right]; \quad \hat C_K=\left[\begin{array}{cc}-1.4044 &
0.3008 \\ -0.2639 & 0.3526
\end{array}\right],
\end{eqnarray*}
and with the LQG cost $J_{\infty}=4.1793$.

\begin{remark}
It is important to note here that although the type of controller
is not a priori specified, the algorithm returns $\Theta_K$ which
corresponds (up to a similarity transformation) to a fully quantum
controller, rather than a classical controller.
\end{remark}

Since $\Theta_K=SJS\trp$ with
$S=\sqrt{0.182}\left[\begin{array}{cc} 0 & 1 \\ 1 &
0\end{array}\right]$, then according to Lemma \ref{lem:pb2-to-pb1}
we find the matrices $A_K$, $B_{Ki}$ ($i=1,2,3$) and $C_K$ solving
Problem \ref{pb:lqg-form} as follows:
\begin{eqnarray}
&&A_K=\left[\begin{array}{cc} 0.0257 & -0.3789
\\ 0.0666 & -0.2125\end{array}\right];
\quad B_{K1}=\left[\begin{array}{cc} 0.1126 & -0.5992 \\
0.1504 & -0.1283 \end{array}\right]; \nonumber \\
&&\quad
B_{K2}=10^{-10} \left[\begin{array}{cc} -0.2721 & 0.0272 \\
-0.1096 & 0.0601
\end{array}\right];
\quad B_{K3}=\left[\begin{array}{cc}1.0297 & -0.1974 \\
0.8255 & -0.0503
\end{array}\right];\nonumber\\
&&\quad C_K=\left[\begin{array}{cc} 0.1283 & -0.5992 \\ 0.1504 &
-0.1126 \end{array}\right]. \label{eq:eg-cntrl-2}
\end{eqnarray}

Note again that $B_{K2}$ has negligibly small entries and in a similar fashion to
 the related discussion in Section 8.1, we may set $B_{K2}=0$ to
obtain a simpler controller. For this controller we have that the
right hand side of (\ref{eq:eq1-ccr-num}) takes the numerical
value $10^{-13}\left[\begin{array}{cc} 0 & 0.6098 \\ -0.6098 &
0\end{array} \right]$. The fact that we were ultimately able to
find a fully quantum controller that outperforms the best
performance of a classical linear controller with indirect
measurement (with $J_{\infty}=4.1793$ in the former compared to
$J_{\infty} \approx 4.4440$ in the latter case) is very
interesting and indicates that the quantum LQG controller should
be considered to be more than a theoretical curiosity, but one
which may potentially be of practical significance.

It is important to note that although the controller
(\ref{eq:eg-cntrl-2}) beats the class of classical linear
controllers with indirect measurements, it may not beat all
possible classical controllers. For example, there may exist
classical non-linear controllers that could perform just as well
or better.

\begin{remark}
Here we do not discuss the physical realizations of fully quantum
controllers such as (\ref{eq:eg-cntrl-1}) and
(\ref{eq:eg-cntrl-2}) in the laboratory using currently available
quantum optical components. This is addressed in the work of \cite{NJD08} on a network synthesis theory for dynamical
quantum optical networks that parallels that of electrical network
synthesis theory.
\end{remark}

\section{Conclusions}
\label{sec:conclude} In this paper, we have formulated a quantum
LQG problem that allows the possibility for the controller to be
another quantum system. In general, this problem is a polynomial
matrix programming problem that can be systematically converted to
a rank constrained LMI problem. To solve the problem, we propose a
numerical procedure based on an alternating projections algorithm
and an extension of this procedure for the case where the type of
controller is not a priori specified. In two examples, we consider
the problem of stabilization of a marginally stable quantum plant
and successfully compute some fully quantum LQG controllers that
achieve this goal. We show that for the fully classical controller
schemes considered herein (in which continuous measurements are
performed on the output of the plant), a fully quantum LQG
controller can be found which gives an improved level of
performance.

It is important to emphasize that the quantum LQG problem posed
here cannot be solved by simply applying the standard LQG
methodology, since the resulting controller may not be physically
realizable (i.e., the controller system matrices are not
guaranteed to satisfy (\ref{eq:cntrl-pr-1}) and
(\ref{eq:cntrl-pr-2})). In fact, unlike the classical LQG problem
that can be reformulated as a convex LMI problem, the quantum LQG
problem is in general computationally hard.

For future investigation, theoretical aspects of the quantum LQG
problem posed herein deserve further study. At the moment, system
theoretic conditions for the existence of a solution to this
problem are not known. Another avenue which is of importance for
future investigation is the practical implementation of quantum
LQG controllers and, indeed, general linear quantum controllers in
the laboratory.

\begin{acknowledgement}                               
The first author thanks R. Orsi for discussions on rank
constrained LMI problems and the LMIRank software, V. P. Belavkin
for discussions on quantum linear systems and dynamic programming,
and J. L\"{o}fberg for correspondences on the Yalmip software.
\end{acknowledgement}

\bibliographystyle{plain}
\bibliography{ieeeabrv,tacbib,rip,mjbib2004,irpnew}

\end{document}